\documentclass[11pt,a4paper]{article}
\usepackage[margin=1in]{geometry}
\usepackage{tikz}
\usetikzlibrary{positioning,arrows.meta}
\usepackage{subcaption}

\usepackage[english]{babel}
\usepackage{algpseudocode}
\usepackage{makecell}
\usepackage[english]{babel}
\usepackage{array} 
\usepackage{booktabs} 
\usepackage{authblk} 

\usepackage[ruled,linesnumbered,vlined]{algorithm2e}
\usepackage{algpseudocode}
\usepackage{amssymb}
\usepackage{mathtools}

\usepackage{amsmath}
\usepackage{amsthm}
\usepackage{graphicx}
\usepackage[colorlinks=true, allcolors=blue]{hyperref}
\usepackage{natbib}
\usepackage{framed}
\definecolor{shadecolor}{gray}{0.9}
\usepackage{pgfplots}
\pgfplotsset{compat=1.17}
\usepackage{tikz}
 \usetikzlibrary{patterns}
\usepackage{pgfplotstable}
    \pgfplotsset{
    name nodes near coords/.style={
        every node near coord/.append style={
            name=#1-\coordindex,
            alias=#1-last,
        },
    },
    name nodes near coords/.default=coordnode
    }

\usepackage{multirow}

\usepackage[table]{xcolor} 

\newtheorem{theorem}{Theorem}[section]
\newtheorem{definition}[theorem]{Definition}

\newtheorem{corollary}[theorem]{Corollary}
\newtheorem{lemma}[theorem]{Lemma}

\newtheorem{observation}[theorem]{Observation}

\newcommand{\E}{\mathbb{E}}
\newcommand{\bx}{\mathbf{x}}

\newcommand{\bR}{{\mathbb{R}}}
\newcommand{\bv}{\mathbf{v}}
\newcommand{\bw}{\mathbf{w}}

\newcommand{\xiaowei}[1]{{\color{red}[XW: #1]}}

\definecolor{ruleblue}{RGB}{70,120,200}
\definecolor{shadegray}{gray}{0.9}

\makeatletter
\renewcommand{\and}{\end{tabular}\hspace{3em}\begin{tabular}[t]{c}}
\makeatother

\title{Fair Division by Contribution: A Shapley Value Perspective\thanks{The authors are ordered alphabetically.}}

\author[1]{Xiaohui Bei}
\author[2]{Pinyan Lu}
\author[3]{Xiaowei Wu}
\author[1]{Shengwei Zhou}

\affil[1]{Nanyang Technological University, \texttt{\{xhbei,shengwei.zhou\}@ntu.edu.sg}}
\affil[2]{Shanghai University of Finance and Economics, \texttt{lu.pinyan@mail.shufe.edu.cn}}
\affil[3]{University of Macau, \texttt{xiaoweiwu@um.edu.mo}}

\begin{document}

\maketitle
\begin{abstract}
    In many resource allocation problems, agents' valuations are best interpreted not as subjective preferences, but as the value they generate from receiving resources. Such valuations capture productivity, effectiveness, or technology, and may differ significantly across agents. In these settings, classical fairness notions such as proportionality or envy-freeness fail to reflect agents' heterogeneous contributions to the collective outcome. Motivated by this perspective, we introduce \emph{Shapley Value Fairness (SVF)} for the allocation of divisible goods without monetary transfers. SVF interprets an agent's entitlement as her expected marginal contribution to optimal social welfare, and uses the Shapley value of the associated welfare maximization game as a normative fairness benchmark. We position SVF relative to existing fairness notions and show that it provides a natural bridge between fairness and efficiency in contribution-based environments.

    Since exact implementation of the Shapley value is generally infeasible without transfers, SVF naturally leads to the problem of finding allocations that approximate this benchmark as well as possible. We provide a systematic worst-case analysis of the achievable Shapley approximation ratio. For general concave valuations, we establish a tight $\Theta(\ln n)$ bound. For capped concave valuations with bounded demands, this bound improves to $\Theta(\ln D)$, where $D$ is the maximum aggregate demand for any item. For linear valuations, we further refine the bound to $\Theta(\min\{k, \ln \gamma, \ln n\})$ in terms of the number of agent types $k$ and the value fluctuation ratio $\gamma$, and show that all bounds are asymptotically tight.

    Beyond worst-case guarantees, we study per-instance and computational aspects of SVF. For concave valuations, we show that the optimal Shapley approximation ratio can be efficiently estimated via sampling, and that near-optimal SVF allocations can be computed in polynomial time with high probability. For linear valuations, both the Shapley values and the optimal SVF allocation are exactly computable in polynomial time. Together, these results establish SVF as a conceptually grounded and computationally tractable fairness notion for resource allocation without monetary transfers.
\end{abstract}

\newpage

\section{Introduction}\label{sec:introduction}
Fair division lies at the heart of economics and algorithmic game theory with a central question: How should resources be distributed among heterogeneous agents in a way that is both \emph{fair} and \emph{efficient}? A key modeling choice concerns how agents' valuations are defined. In the classical fair division literature, an agent's valuation reflects how much she likes a bundle of resources. Such preferences are subjective and only meaningful up to positive scaling, which makes interpersonal utility comparison problematic. Classic fairness notions such as \emph{proportionality}~\cite{steihaus1948problem} and \emph{envy-freeness}~\cite{foley1967resource} fit naturally in this setting, as they are scale-free and do not rely on comparing utility levels across agents. In many applications, however, valuations are better interpreted as the value \emph{generated} by allocating resources to an agent. Agents' valuations capture productivity, effectiveness, or technology for converting resources into social value. Examples include research funding allocation, public project investment, or computational resource sharing, where agents differ objectively in their ability to produce impact. Under this interpretation, valuations are not normalized and should not be scale-invariant; social welfare, defined as the sum of agents' generated values, is a naturally meaningful objective. This paper focuses on this latter setting, and investigates how fairness can be defined and achieved that explicitly accounts for agents' heterogeneous contributions to total value.

Let us consider an illustrative scenario of research budget allocation: a funding agency must distribute a fixed budget among several research proposals. Each proposal has a cap on the maximum funding it can effectively utilize, as well as a valuation function representing its potential societal impact, and these valuations could differ significantly in scale.
More concretely, consider the following example in Figure~\ref{fig:research_budget} with five proposals, each asking for up to 100,000 dollars, and a total budget of 200,000 dollars. 
Let us further assume for simplicity that the valuation functions are linear in the amount of funding it receives, capped at its maximum.

\begin{figure}[!ht]
\centering
    \begin{tabular}{|c|ccccc|c|}
    \hline
    Proposal & A & B & C & D & E & Budget \\
    \hline
    Valuation (per unit of funding) & 10 & 3.1 & 3 & 2 & 1 & \\
    \hline  
    Max Funding & 100 & 100 & 100 & 100 & 100 & 200\\
    \hline
    \multicolumn{7}{l}{Allocations:} \\
    \hline
    Equal Treatment & 40 & 40 & 40 & 40 & 40 &  \\
    Max-Min Fair & 8.9 & 28.6 & 29.6 & 44.3 & 88.6 &  \\
    Weighted Fair & 104.7 & 32.5 & 31.4 & 20.9 & 10.5 &  \\
    Social Welfare & 100 & 100 & 0 & 0 & 0 &  \\
    Shapley Value Fair & 61.9 & 38.9 & 37.8 & 32.5 & 28.9 &  \\
    \hline
    \end{tabular}
    \caption{An example of research budget allocation with five proposals. Units are in thousands of dollars.}
    \label{fig:research_budget}
\end{figure}

What would a fair allocation look like in this scenario? Classic fairness notions, such as proportionality or envy-freeness, would treat all proposals equally and mandate an equal split of the budget, giving each proposal 40,000 dollars.
However, this allocation ignores the differences in the proposals' societal value, and thus is clearly inefficient and arguably unfair, as high-impact proposals are underfunded while low-impact ones receive more than necessary.
Another classic fairness notion is based on Rawlsian \emph{max-min fairness}~\cite{rawls1971theory}, which would allocate the budget to maximize the minimum valuation across proposals. This would lead to an even more inefficient allocation, as it would prioritize the lowest-value proposal (E) at the expense of overall welfare. Alternatively, one might consider a weighted fairness approach, where the funding each proposal receives is proportional to its valuation. While this method accounts for differences in valuation, it is somewhat inflexible and cannot generalize to settings beyond linear valuations. For instance, in this example, proposal A would receive $200 \times \frac{10}{10+3.1+3+2+1} \approx 104.7$ thousand dollars, exceeding its maximum effective funding.
Lastly, one could adopt a completely utilitarian approach and allocate the budget to maximize total societal impact. This means fully funding proposals A and B, yielding a total valuation of $1,310$. While efficient, this allocation is blatantly unfair. For example, proposals B and C have an almost identical valuation per dollar, yet C receives nothing while B is fully funded.

A more natural way to incorporate contribution into fairness comes from cooperative game theory, where agents form coalitions to achieve collective outcomes.
How to fairly divide the surplus generated by a coalition among its members in such a game has been extensively studied. Most notably, the \emph{Shapley value}~\cite{shapley1953value}, defined as each agent's expected marginal contribution to the total surplus when agents join a coalition in random order, is the de facto standard for surplus sharing, providing a canonical bridge between fairness and efficiency.

The idea of using Shapley value as a fairness metric in resource allocation problems is not new. In a seminal paper, Moulin~\cite{moulin1992application} modeled the resource allocation problem as a cooperative game, where the value of each coalition is given by the maximum social welfare that coalition can achieve from the best allocation of the full resource among the members of that coalition. Moulin argued that in a divisible resource allocation problem, and when monetary transfers are allowed, one should adopt an efficient (i.e., social welfare maximizing) allocation rule and assign agents their Shapley value share of the total welfare. Moulin showed that under a substitutability assumption, which is related to the concavity of valuation functions, this fairness notion satisfies several strong axioms, such as efficiency, proportionality, resource and population monotonicity.

However, fair division often occurs and is more challenging in settings when monetary transfers are not allowed. 
In the research budget allocation scenario described earlier, the funded projects cannot transfer money between each other after the fund allocation.
From a game-theoretic perspective, this corresponds naturally to a non-transferable utility (NTU) setting, where utility cannot be redistributed across agents, and feasible outcomes are constrained by the underlying allocation of resources.
In such contexts, the specific utility profile prescribed by the Shapley value is often unattainable.
In other words, while the sum of agents' Shapley values equals the optimal social welfare, there may not exist an allocation of the items that simultaneously grants every agent their Shapley utility.
Consequently, the Shapley value shifts from being a direct allocation rule to serving as a \emph{fairness benchmark}, which we call \emph{Shapley value fairness} (SVF).

Under this SVF framework, a natural goal would be to seek allocations that approximate the Shapley value for each agent as closely as possible. 
Formally, an allocation is said to achieve an $\alpha$-approximation of SVF if every agent receives at least a $1/\alpha$ fraction of their Shapley value, and we are looking for allocations that minimize this approximation factor $\alpha$.
Let us return to the research budget allocation example in Figure~\ref{fig:research_budget} to illustrate this concept. One can compute and verify that the Shapley values for proposals A, B, C, D, and E are $2570/3$, $500/3$, $470/3$, $90$, $40$, respectively. 
Note that these Shapley values sum to $1310$, which is the optimal social welfare achievable with the given budget. Without monetary transfers, however, it is impossible to allocate funds so that each proposal obtains its full Shapley value. 
In this case, the allocation that gives the best approximation of SVF is to fund proposals A, B, C, D, and E by an amount of $61.9$, $38.9$, $37.8$, $32.5$, and $28.9$ thousand dollars, respectively, which guarantees each proposal at least $72.2$\% of its Shapley value. 
Note that this allocation is significantly more efficient than both the proportional and max-min fair allocations discussed earlier, achieving a total welfare of $946.89$ compared to $764$ and $443.66$, respectively. 
It is also fairer than the utilitarian allocation, as every proposal receives a non-trivial fraction of its Shapley value. 
Notably, proposals B and C receive similar funding under this allocation, which reflects their comparable valuation per dollar.

The Shapley value benchmark aligns with our normative goal, that agents are treated fairly precisely to the extent that they create value for the collective.
Using the Shapley value as a fairness threshold is therefore both natural and principled: it respects agents' heterogeneous productivity while maintaining an egalitarian interpretation.
At the same time, when exact Shapley value allocations are unattainable, seeking the best possible approximation of SVF leads to a natural balance between fairness (in terms of Shapley value) and efficiency (in terms of total social welfare). 
One can show that any allocation that achieves an $\alpha$-approximation of SVF will guarantee simultaneously both an $\alpha$ fraction of the optimal social welfare and $\alpha$-proportionality for all agents (see Lemma~\ref{lemma:proportionality_and_efficiency}). 
This unifies fairness and efficiency in one framework with a single objective parameter $\alpha$ to optimize for the decision maker, thereby simplifying the usually conflicting trade-off between these two desiderata.

\subsection{Our Results}

In this paper, we initiate a systematic study of \emph{Shapley Value Fairness} (SVF) in the context of divisible resources without monetary transfers. Our starting point is normative: we argue that the Shapley value provides a natural and principled fairness benchmark when agents differ in their ability to generate social value. SVF formalizes the idea that agents should be treated fairly \emph{in proportion to their contribution} to optimal social welfare. Taking into consideration that the exact implementation of the Shapley value may be infeasible without monetary transfers, this naturally leads to the study of approximate SVF allocations, which capture the best achievable compromise between fairness and efficiency.

Our first set of results concerns the limits of approximating the Shapley value under concave valuations. We show that the optimal Shapley approximation ratio is $\Theta(\ln n)$ in the worst case, and the bound is tight: we provide an allocation algorithm that achieves this guarantee, as well as matching lower bounds showing that no asymptotically better approximation is possible in general. This result should be interpreted not merely as an approximation guarantee, but as a structural statement about the inherent tension between contribution-based fairness and feasibility in large systems. 
We also consider an important subclass of \emph{capped} concave valuations, where each agent has a maximum demand $d_i(e)$ for item $e$, and $D = \max_{e\in M} \sum_{i\in N} d_i(e)$ represents the maximum aggregate demand for any single item, assuming that each item has a unit supply.
We obtain a refined worst-case bound of $\Theta(\ln D)$. Note that $D$ can be much smaller than $n$ in many applications. This result demonstrates that SVF can be substantially better approximated in realistic allocation settings, such as budget allocation problems where individual demands are limited.
Beyond worst-case analysis, we also study \emph{per-instance guarantees} and computational aspects of SVF. We show that the optimal Shapley approximation ratio can be efficiently estimated using a polynomial number of samples, and that a near-optimal SVF allocation can be computed with high probability via convex optimization. 

We further investigate the domain of \emph{linear} valuations, which allows for a sharper analysis. In this setting, we show that the exact Shapley value of all agents, as well as the optimal Shapley approximation allocation, can both be computed in polynomial time. Moreover, we refine the logarithmic approximation bound by identifying meaningful structural parameters. In particular, we show that the approximation ratio improves to $\Theta(\min\{k, \ln \gamma, \ln n\})$, where $k$ is the number of distinct agent types and $\gamma$ is the ratio between the maximum and minimum values. These refinements are especially relevant in practice, as real-world instances often exhibit limited heterogeneity or bounded value dispersion.

\subsection{Related Works}
Axiomatic approaches to fair division have been extensively studied in the literature~\cite{books/daglib/0017734}. 
Proportionality and envy-freeness are among the most fundamental fairness axioms, and have motivated a wide range of allocation rules in both divisible and indivisible resource settings. 
Notable examples include the probabilistic serial rule~\cite{bogomolnaia2001new}, the maximum Nash social welfare allocation~\cite{varian1973equity}, and the Pareto efficient egalitarian equivalent allocation (PEEEA)~\cite{pazner1978egalitarian}, each satisfying different subsets of these axioms while exhibiting distinct efficiency and incentive trade-offs.

Within axiomatic approaches to fair allocation, marginal contribution principles provide an alternative perspective to proportional and egalitarian criteria. 
However, despite its fundamental role in cooperative game theory, the Shapley value has received comparatively little attention in the fair division literature. 
Moulin's work~\cite{moulin1992application} provides one of the few axiomatic treatments of the Shapley value in this context. 
He highlighted the application of Shapley value to the divisible resource allocation problem with monetary transfers, with the only rule that satisfies all axioms alongside efficiency at the same time.
By viewing the allocation process as a cooperative game, our work is also closely related to several classes of allocation games, including assignment games~\cite{shapley1971assignment}, matching games~\cite{journals/mor/DengIN99,conf/stacs/AzizK14}, and cost-sharing games~\cite{moulin1992serial,moulin2001strategyproof}. 
The Shapley value has additionally been used as a fairness benchmark in auction and market design settings~\cite{journals/corr/abs-2410-18602}.
Beyond theoretical studies, the Shapley value has also been proposed for research funding allocation in national evaluation systems. 
Demetrescu et al.~\cite{demetrescu2019shapley} studied its application in the Italian VQR research assessment exercise, where evaluation outcomes directly determine institutional funding.

Recent work also studies fair allocation of indivisible items with \emph{social impact}, where each agent is associated with two distinct valuation functions: a private preference and a separate measure of social impact or externality. 
These works investigate notions of fairness that account for both objectives simultaneously~\cite{bu2023fair}, or aim to maximize social impact subject to fairness constraints~\cite{deligkas2025dividing,flammini2025fair}. 
In contrast, our setting does not distinguish between private utility and social impact; instead, agents' valuations directly represent their contribution to social welfare.

Finally, a closely related line of work studies \emph{weighted fair division}, where agents are associated with unequal entitlements represented by exogenously given weights. Classical extensions of proportionality and envy-freeness to this setting arise naturally in cake-cutting and resource allocation with heterogeneous claims \cite{robertson1998cake}. More recently, weighted variants of fairness notions have been studied extensively in the allocation of indivisible items, including weighted envy-freeness and its relaxations such as weighted envy-freeness up to one item (WEF1) \cite{chakraborty2021weighted,chakraborty2024weighted}, as well as existence and structural results in weighted allocation models \cite{manurangsi2025asymptotic,suksompong2025weighted}. 
In these works, weights are typically specified externally and interpreted as priorities or entitlements. In contrast, our framework derives weights endogenously from agents' marginal contributions to social welfare via the Shapley value.

\section{Preliminaries}\label{sec:preliminaries}

We consider the problem of allocating $m$ divisible items $M$ to a set of $n$ agents $N$.
We represent the total available resources as the vector $\mathbf{1} = (1, \ldots, 1) \in [0,1]^m$.
A \emph{bundle} is a vector $x \in [0,1]^m$, where the component $x(e)$ represents the fraction of item $e \in M$ contained in the bundle.
Each agent $i \in N$ possesses a valuation function $v_i: [0,1]^m \to \mathbb{R}_{\ge 0}$ over bundles.
We assume that valuations satisfy $v_i(\mathbf{0}) = 0$ for all $i \in N$.
We denote the profile of valuation functions by $\bv = (v_1, \ldots, v_n)$.

An \emph{allocation} $\bx = (x_1, \ldots, x_n)$ is a collection of bundles, where $x_i$ denotes the bundle allocated to agent $i$.
An allocation is feasible if $\sum_{i\in N} x_i(e) = 1$ holds for every item $e \in M$.
The \emph{social welfare} of an allocation $\bx$ is the sum of the agents' utilities, $\sum_{i\in N} v_i(x_i)$.
Next, we define the \emph{group welfare} function $f: 2^N \times [0,1]^m \to \mathbb{R}_{\ge 0}$.
For any subset of agents $S \subseteq N$ and any resource vector $y \in [0,1]^m$, the group welfare $f(S, y)$ is the maximum social welfare achievable by distributing the resources $y$ among the members of $S$:
\begin{equation}
    f(S, y) = \max \left\{ \sum_{i \in S} v_i(x_i) : \sum_{i \in S} x_i = y; \; x_i \in [0,1]^m, \forall i \in S \right\}.
\end{equation}

When the resource vector is omitted, $f(S)$ refers to the welfare given the full set of items, i.e., $f(S) = f(S, \mathbf{1})$.
For any positive integer $k$, we use $[k]$ to denote the set $\{1, 2, \ldots, k\}$.

\subsection{Shapley Value Fairness}

We adopt the classic solution concept from cooperative game theory: the Shapley value.
The underlying cooperative game is defined by the characteristic function $f(S)$, which represents the maximum welfare coalition $S$ can achieve with the full resource vector $\mathbf{1}$.

\begin{definition}[Shapley Value]
    The Shapley value $\phi_i$ of agent $i$ is defined as the expected marginal contribution of agent $i$ to the group welfare under a uniform random permutation of agents.
    Let $\pi: N \to [n]$ be a random permutation. We define:
    \begin{equation}
        \phi_i = \E_{\pi} \left[ f(S_i^\pi \cup \{i\}) - f(S_i^\pi) \right],
    \end{equation}
    where $S_i^\pi = \{ j \in N : \pi(j) < \pi(i) \}$ is the set of predecessors of $i$ under permutation $\pi$.
\end{definition}

In the context of cooperative games with transferable utility (money), Moulin~\cite{moulin1992application} demonstrated that it is always possible to achieve an allocation where every agent receives utility exactly equal to their Shapley value.
This is achieved through monetary transfers that redistribute the surplus.
However, we consider the setting \emph{without} monetary transfers. In this non-transferable utility domain, allocations satisfying $v_i(x_i) \ge \phi_i$ for all agents $i\in N$ may not exist.
This motivates us to consider multiplicative approximations.

\begin{definition}[Shapley Value Fairness]
    For a value $\alpha \ge 1$, an allocation $\bx$ is $\alpha$-approximate Shapley Value Fair ($\alpha$-SVF) if for any agent $i \in N$:
    \begin{equation}
        v_i(x_i) \geq \frac{\phi_i}{\alpha}.
    \end{equation}
    When $\alpha=1$, the allocation satisfies exact Shapley Value Fairness.
\end{definition}

We define the \emph{Shapley approximation ratio} $\alpha(\bx)$ of a feasible allocation $\bx$ as the maximum ratio of the Shapley value to the allocated utility across all agents:
\begin{equation}
    \alpha(\bx) = \max_{i \in N} \frac{\phi_i}{v_i(x_i)}.
\end{equation}

We define our primary allocation rule as the mechanism that finds the strongest possible fairness guarantee (i.e., the smallest possible $\alpha$).

\begin{definition}[Optimal Shapley Approximation Rule]
    The Optimal Shapley Approximation rule selects a feasible allocation $\bx^*$ that minimizes the approximation factor $\alpha$:
    \begin{equation}
        \bx^* \in \arg\min_{\bx \in \mathcal{X}} \alpha(\bx),
    \end{equation}
    where $\mathcal{X}$ denotes the set of all feasible allocations.
    We call $\alpha^* = \alpha(\bx^*)$ the optimal Shapley approximation ratio.
\end{definition}

\subsection{Relation to Other Fairness Notions}
It is instructive to compare Shapley Value Fairness with other well-known fairness notions in the fair division literature.
Two of the most prominent notions are proportionality and envy-freeness.

\begin{definition}[Proportionality]
    An allocation $\bx$ is said to be \emph{proportional} if for all $i\in N$, it holds that $v_i(x_i) \ge v_i(\mathbf{1})/{n}$ .
    It is said to be $\alpha$-proportional for some $\alpha \ge 1$ if for all $i \in N$, it holds that $v_i(x_i) \ge v_i(\mathbf{1})/(\alpha \cdot n)$.
\end{definition}

\begin{definition}[Envy-freeness]
    An allocation $\bx$ is said to be \emph{envy-free} if for any pair of agents $i,j \in N$, it holds that $v_i(x_i) \ge v_i(x_j)$.
    It is said to be $\alpha$-envy-free for some $\alpha \ge 1$ if for any pair of agents $i,j \in N$, it holds that $v_i(x_i) \ge v_i(x_j) / \alpha$.
\end{definition}

\paragraph{Note on Normalization.}
Both proportionality and envy-freeness are scale-invariant fairness notions, meaning that they are not affected by positive scaling of an agent's valuation function.
As a result, standard fair division models often assume valuation functions are \emph{normalized} by setting $v_i(\mathbf{1}) = 1$ for all agents $i \in N$.
Such normalization is commonly used to formalize homogeneity among agents and remove scale effects in agents' subjective valuations.
However, this work focuses on settings where agents' valuations reflect productivity or their technologies to convert resources into value. Such valuations are objective and \emph{should not} be normalized away. We thereby do not impose any normalization on agents' valuation functions. Consequently, the SVF notion is \emph{not} scale-invariant, which makes it distinct from other scale-invariant fairness notions like proportionality, envy-freeness, or the Nash welfare rule.

\medskip

As we discussed in the introduction, Moulin~\cite{moulin1992application} demonstrated that the Shapley value, as a fairness rule, satisfies proportionality and efficiency.
The following lemma shows the robustness of these two properties when extending to the approximation of Shapley value fairness.

\begin{lemma}\label{lemma:proportionality_and_efficiency}
    Any $\alpha$-SVF allocation $\bx$ is both $\alpha$-proportional and an $\alpha$-approximation to the optimal social welfare.
\end{lemma}
\begin{proof}
    Recall that the Shapley value $\phi_i$ represents the expected marginal contribution of agent $i$ over a uniformly random permutation $\pi$.
    With probability $1/n$, agent $i$ appears at the first position in the permutation.
    Hence, we have the lower bound of the Shapley value as $\phi_i \ge v_i(\mathbf{1}) / n$.
    Since the allocation $\bx$ is $\alpha$-SVF, we have $v_i(x_i) \ge \phi_i / \alpha \ge v_i(\mathbf{1}) / (\alpha n)$.
    This is precisely the definition of an $\alpha$-proportional allocation.
    
    The social welfare guarantee follows directly from the efficiency of the Shapley value.
    A fundamental property of Shapley values is that they partition the total value of the grand coalition:
    \begin{equation*}
        \sum_{i \in N} \phi_i = f(N).
    \end{equation*}
    
    Note that $f(N)$ is exactly the optimal social welfare.
    Summing the $\alpha$-SSF guarantee $v_i(x_i) \ge \phi_i / \alpha$ over all agents $i \in N$, we obtain:
    \begin{equation*}
        \sum_{i \in N} v_i(x_i) \ge \sum_{i \in N} \frac{\phi_i}{\alpha} = \frac{1}{\alpha} \sum_{i \in N} \phi_i = \frac{f(N)}{\alpha}.
    \end{equation*}
    Thus, the allocation achieves at least a $1/\alpha$ fraction of the optimal social welfare.
\end{proof}

On the other hand, the following lemma shows that Shapley value fairness is generally incomparable to envy-freeness.

\begin{lemma}\label{lemma:envy-free}
    Any $\alpha$-SVF allocation $\bx$ is $\alpha n$-envy-free.
    Moreover, there exists an instance and an $\alpha$-SVF allocation with $\alpha \geq 2$ with an approximation ratio of at least $(\alpha-1)n$ to envy-freeness.
\end{lemma}
\begin{proof}
    Similar to the previous analysis, we have $v_i(x_i) \ge \phi_i/\alpha \ge v_i(\mathbf{1}) / (\alpha n)$.
    Since for any $j\neq i$, we have $v_i(x_j) \le v_i(\mathbf{1})$, the allocation is naturally $\alpha n$-envy-free.

    We show that the linear dependence on $n$ is unavoidable.

    Consider an instance of allocating a single item $e$ to $n$ agents with linear functions.
    The valuation function of agent $1$ is $v_1(x) = \beta \cdot x$, where $\beta \gg n$; for any agent $i\geq 2$, her valuation function is $v_i(x) = x$.
    Agent $1$ has a marginal contribution of at least $\beta-1$ regardless of her position in the random permutation, while other agents have non-zero contributions only if they are in the first position.
    Therefore we have $\phi_1 = \beta - \frac{n-1}{n}$ and $\phi_i = \frac{1}{n}$ for any $i\geq 2$.
    Consequently, the optimal Shapley approximation ratio is
    \begin{equation*}
        \sum_{i=1}^n \frac{\phi_i}{v_i(1)} = 1 - \frac{n-1}{\beta\cdot n} + \frac{n-1}{n} \approx 2 \quad (\text{when $\beta,n \to \infty$}). 
    \end{equation*}
    
    For any $\alpha \geq 2$, setting $x_i = \frac{1}{\alpha\cdot n}$ and $x_1 = 1-\frac{n-1}{\alpha\cdot n}$ gives a $\alpha$-SVF allocation, in which agent $i$ is $((\alpha-1)n+1)$-approximate envy-free towards agent $1$.
\end{proof}

\paragraph{Weighted Fairness Notions.}
Another important class of fairness notions has to do with \emph{weighted} fairness, such as weighted proportionality and weighted envy-freeness~\cite{robertson1998cake,zeng2000approximate}.
These notions assign each agent a weight that reflects their entitlement to the resources.

\begin{definition}[Weighted Fairness Notions]
    Given a weight vector $\bw = (w_1, \ldots, w_n)$ with $w_i > 0$ for all $i \in N$, an allocation $\bx$ is:
    \begin{itemize}
        \item \emph{weighted proportional (WPROP)} if for any $i \in N$, it holds that $v_i(x_i) \ge \frac{w_i}{\sum_{j \in N} w_j} v_i(\mathbf{1})$.
        \item \emph{weighted envy-free (WEF)} if for any pair of agents $i,j \in N$, it holds that $\frac{v_i(x_i)}{w_i} \ge \frac{v_i(x_j)}{w_j}$.
    \end{itemize}
    Moreover, an allocation $\bx$ is said to be $\alpha$-WPROP for some $\alpha \ge 1$ if for all $i\in N$, it holds that $v_i(x_i) \ge \frac{w_i}{\alpha \cdot \sum_{j\in N} w_j} v_i(\mathbf{1})$.
\end{definition}

Conceptually, Shapley value fairness offers a compelling alternative to classic weighted fairness.
Unlike traditional schemes where weights are assigned externally, the Shapley value \emph{derives} weights endogenously from agents' own productivity and their marginal contributions to group welfare.
This connection highlights the flexibility of Shapley Value Fairness in capturing diverse notions of entitlement among agents.

With this being said, if we set the weights as the Shapley values, i.e., $w_i = \phi_i$ for all $i \in N$, then an $\alpha$-SVF allocation can be compared with weighted fairness notions. We have the following lemma.

\begin{lemma}\label{lemma:weighted_proportionality}
    An $\alpha$-SVF allocation $\bx$ is also $\alpha$-WPROP if $w_i = \phi_i$ for any agent $i\in N$.
\end{lemma}
\begin{proof}
    Let the weight of agent $i$ be $w_i = \phi_i$.
    By the efficiency axiom of the Shapley value, the sum of weights equals the value of the grand coalition: $\sum_{j \in N} \phi_j = f(N)$.
    Thus, the weighted proportional share for agent $i$ is:
    \begin{equation*}
        \frac{w_i}{\sum_{j \in N} w_j}\cdot v_i(\mathbf{1}) = \frac{\phi_i}{f(N)} \cdot v_i(\mathbf{1}).
    \end{equation*}
    
    By the definition of the coalitional value function, $f(N)$ represents the maximum social welfare, which implies $f(N) \ge v_i(\mathbf{1})$.
    Consequently, the weighted proportional share of agent $i$ is at most $\phi_i$.
    Hence any $\alpha$-SVF allocation guarantees $\alpha$-WPROP.
\end{proof}


\section{Shapley Value Fairness for Concave Valuations}
\label{sec:general_concave}


Having introduced the Optimal Shapley Approximation Rule, it is immediately natural to ask two questions: (1) Approximation guarantee: What is the worst-case approximation ratio achievable by the optimal Shapley approximation rule for a given domain of valuation functions? (2) Computation: Can we efficiently compute the optimal Shapley approximation ratio and the corresponding allocation for a given resource allocation instance?

Following the framework of Moulin~\cite{moulin1992application}, in this section, we provide answers to both questions above for the domain of \emph{concave} valuations.
A valuation function $v: [0,1]^m \to \bR_{\ge 0}$ is concave if for any two bundles $x, y \in [0,1]^m$ and any $\lambda \in [0,1]$, we have 
$$
v(\lambda x + (1-\lambda) y) \ge \lambda\cdot v(x) + (1-\lambda)\cdot v(y).
$$
We will assume that for every agent $i$, the function $v_i$ is concave, non-decreasing, and continuous.

\medskip

We start by discussing the approximation guarantee of the optimal Shapley approximation rule.
Our main result establishes that with concave valuations, this approximation guarantee is $O(\ln n)$, and this bound is asymptotically tight.
This is shown through a constructive algorithm and a matching impossibility result.

\subsection{Logarithmic Approximation Algorithm for Concave Valuations}
\label{ssec:algorithm_concave}

We first show that the optimal Shapley approximation ratio is at most $O(\ln n)$ for all concave valuations by defining an allocation rule.

\begin{theorem}[Upper Bound] \label{theorem:concave_upper_bound}
    For any instance with concave valuations, there exists a feasible allocation that is $(\ln n + 1)$-SVF.
\end{theorem}

We construct the allocation based on the underlying structure of the Shapley value.
To do so, we first establish a useful upper bound on the Shapley value of each agent.
We first provide an upper bound on the marginal contribution of an agent $i$ to a group $S$ by the value of the allocated resource to agent $i$ upon its arrival.
For any subset $S \subseteq N$, we let $\bx(S)$ denote the allocation vector that maximizes social welfare for the coalition $S$.\footnote{If there are multiple social welfare maximizing allocations, we pick one arbitrarily and use this fixed allocation throughout our analysis. In other words, $\bx: 2^N \to [0,1]^N\times M$ is a deterministic function that maps a group $S$ of agents to a unique social welfare maximizing allocation among agents in $S$.}
Let $x_i(S)$ be the specific resources allocated to agent $i\in S$ in this optimal solution.
Now we consider the upper bound of the marginal contribution of agent $i$ to coalition $S$.

\begin{lemma}\label{lemma:concave_marignal_upper_bound}
    For any $S\subseteq N$ and $i\notin S$, we have
    $f(S \cup \{i\}) - f(S) \le v_i(x_i(S \cup \{i\}))$.
\end{lemma}
\begin{proof}
    Note that $x(S\cup \{i\})$ is the optimal allocation for coalition $S \cup \{i\}$.
    By definition, we have
    \begin{equation*}
        f(S\cup \{i\}) = v_i(x_i(S\cup \{i\})) + \sum_{j\in S} v_j(x_j(S\cup \{i\})).
    \end{equation*}
    
    Since $f(S)$ is the optimal welfare for $S$ using the \emph{entire} set of resources, and the group welfare is monotone with respect to resources, the welfare generated by $S$ using partial resources $\sum_{j\in S} x_j(S\cup \{i\})$ is upper-bounded by $f(S)$:
    \begin{align*}
         \sum_{j\in S} v_j(x_j(S\cup \{i\})) \le f(S, \sum_{j\in S} x_j(S\cup \{i\})) \le f(S, \mathbf{1}).
    \end{align*}
    
    Substituting this equality back yields 
    \begin{equation*}
        f(S \cup \{i\}) \le v_i(x_i(S \cup \{i\})) + f(S).
    \end{equation*}
    
    Rearranging the inequality yields the claim.
\end{proof}

Next, we provide an upper bound on the Shapley value.

Recall that $\pi$ is the random permutation of the agents, and $S_i^\pi$ is the set of predecessors of agent $i$ under $\pi$.
By definition, $x_i(S_i^\pi \cup \{i\})$ be the resources allocated to agent $i$ in the optimal social welfare solution for coalition $S_i^\pi \cup \{i\}$.
For convenience of notation, we use $x^\pi_i$ to denote $x_i(S_i^\pi \cup \{i\})$, and call it the \emph{temporary bundle} of agent $i$ under permutation $\pi$.
Note that the above definition is referring to the allocation right after the arrival of agent $i$.
When more agents arrive (after $i$), the social welfare maximizing allocation could change, and the allocated amount of resource for agent $i$ may decrease.
We further define the expected temporary bundle $\hat{x}_i$ for agent $i$ as:
\begin{equation*}
    \hat{x}_i = \E_\pi \left[ x_i^\pi \right].
\end{equation*}

Note that $\hat{x}_i$ is a deterministic vector representing the centroid of the agent's temporary bundles.
We show that each agent's Shapley value is bounded by the value of $\hat{x}_i$.

\begin{lemma}\label{lemma:concave_Shapley_upper_bound}
    The Shapley value of an agent is upper bounded by its value on the expected temporary bundle, i.e., we have $\phi_i \le v_i(\hat{x}_i)$ for all agent $i\in N$.
\end{lemma}
\begin{proof}
    Following the definition of Shapley value, it corresponds to the expectation of agent $i$'s marginal contribution. Hence, we have
    \begin{align*}
        \phi_i &= \E_\pi \left[f(S_i^\pi \cup \{i\}) - f(S_i^\pi)) \right] \\
        &\le \E_\pi \left[ v_i(x_i^\pi) \right]\\
        &\le v_i\left( \E_\pi [x_i^\pi] \right) = v_i(\hat{x}_i),
    \end{align*}
    where the first inequality follows from Lemma~\ref{lemma:concave_marignal_upper_bound}, and the second inequality follows from Jensen's Inequality for concave functions.
\end{proof}

Note that $(\hat{x}_1,\hat{x}_2,\ldots,\hat{x}_n)$ is not necessarily a feasible allocation, because $\hat{x}_i$ describes the expected allocated resource upon agent $i$'s arrival, which could be much higher than the final allocation after all agents have arrived.
Nevertheless, we can show that the summation of $\hat{x}_i$ is upper bounded.
This is important as (intuitively speaking) the summation represents the total resources required to grant every agent their Shapley value simultaneously.

\begin{lemma} \label{lemma:upper_bound_sum_of_x_i(e)}
    For any item $e\in M$, we have $\sum_{i\in N} \hat{x}_i(e) \le \ln n +1$.
\end{lemma}
\begin{proof}
    We proceed by rewriting the sum of expectations as the expectation of the sum, grouping terms by their arrival positions in the random permutation. 
    By definition, $\pi^{-1}(j)$ is the agent arriving in the $j$-th position under permutation $\pi$.
    \begin{align*}
        \sum_{i\in N} \hat{x}_i(e) &= \sum_{i\in N} \E_\pi \left[ x_i^\pi(e) \right] 
        = \E_\pi \left[ \sum_{i\in N} x_{i}^\pi(e) \right] \\
        & = \E_\pi \left[ \sum_{j=1}^n x_{\pi^{-1}(j)}^\pi(e) \right] = \sum_{j=1}^n \E_\pi \left[ x_{\pi^{-1}(j)}^\pi(e) \right].
    \end{align*}
    
    Consider the term $\E_\pi [x_{\pi^{-1}(j)}^\pi(e)]$. 
    This represents the expected fraction of item $e$ allocated to the agent who arrives at position $j$.
    Note that for any fixed position $j$, the set of agents $S_{j} = \{\pi^{-1}(1), \dots, \pi^{-1}(j)\}$ is a random subset of size $j$.
    The social welfare maximizing solution $x(S_j)$ depends only on the set $S_j$, not on the arrival order. Thus, for any set $S_j$, the sum of allocated resources is $\sum_{k \in S_j} x_k(e) = 1$.
    Because the permutation is uniform, the agent at position $j$ (the ``newcomer'') is effectively a uniformly random member of the set $S_j$.
    Therefore, the expected resource allocated to the $j$-th agent, conditioned on the set $S_j$, is the average resource:
    \begin{equation*}
        \E \left[x_{\pi^{-1}(j)}^\pi(e) \mid S_j\right] = \frac{1}{j} \sum_{k \in S_j} x_k(e) = \frac{1}{j}.
    \end{equation*}
    
    Taking the expectation over all sets $S_j$, we have $\E_\pi [x_{\pi^{-1}(j)}^\pi(e)] = 1/j$.
    Summing over all positions yields the lemma: ($H_n = \sum_{j\leq n}1/j$ is the harmonic number)
    \begin{equation*}
        \sum_{i\in N} \hat{x}_i(e) = \sum_{j=1}^n \frac{1}{j} = H_n \le \ln n + 1.
        \qedhere
    \end{equation*}
\end{proof}

It remains to construct an allocation with a Shapley approximation ratio of $(\ln n + 1)$.

Let $\alpha = \ln n + 1$.
By Lemma~\ref{lemma:upper_bound_sum_of_x_i(e)}, we have $\sum_{i\in N} \hat{x}_i(e) \le \alpha = \ln n +1$ for each item $e\in M$.
For each agent $i\in N$ and each item $e\in M$, we allocate
\begin{equation*}
    x_i(e) = \frac{\hat{x}_i(e)}{\sum_{j\in N} \hat{x}_j(e)}
\end{equation*}
fraction of $e$ to $i$.
Clearly, the allocation is feasible as $\sum_{i\in N} x_i(e) = 1$.
Next, we argue that the allocation is $\alpha$-SVF.
Note that for any concave function with $v_i(\mathbf{0})=0$, we have $v_i(\lambda x) \ge \lambda v_i(x)$ for any $\lambda \in [0,1]$.
Then for each agent $i\in N$, we have
\begin{equation*}
    v_i(x_i) \ge v_i \left(\frac{\hat{x}_i}{\alpha}\right) \ge \frac{v_i(\hat{x}_i)}{\alpha} \ge \frac{\phi_i}{\alpha}.
\end{equation*}

In other words, the allocation achieves an $\alpha = \ln n + 1$ approximation of Shapley Value.

\subsection{Capped Concave Valuations}
\label{ssec:capped_concave}

Complementing the general result, we now turn to a specific subclass of concave valuations, which we term \emph{capped concave} valuations.
In this setting, each agent $i\in N$ has a specific demand $d_i(e) \in (0,1]$ for each item $e$, effectively truncating their valuation function when $x_i(e) = d_i(e)$ in the $e$-th dimension.
Specifically, suppose that for each agent $i\in N$, there exists $d_i\in (0,1]^m$ such that $d_i(e)$ represents an \emph{upper bound} on the demand of agent $i$ on resource $e$.
Let $\min\{x,d_i\}$ be defined as a vector with the $e$-th dimension being $\min\{x(e), d_i(e)\}$.
Then we have
\begin{equation*}
    v_i(x) = v_i(\min\{x,d_i\}).
\end{equation*}



A motivating example for studying capped valuation functions arises in research funding allocation.
A funding agency has a total budget of $1$ unit of money to distribute among $n$ researchers, where each researcher requests $d_i$ fraction of the total resource.
Typically, we have $d_i \ll 1$ for all $i\in N$, but the requested total resource is larger than the total budget, i.e., $\sum_{i\in N} d_i > 1$.
Our goal is to decide a feasible, fair, and efficient way to allocate the resources. 
%
This model also captures a wide range of divisible (multi-)resource allocation scenarios with individual demand constraints. 
%

We prove that when agents have capped concave valuations, the tight approximation bound becomes $\Theta(\ln D)$, where $D = \max_{e\in M} \sum_{i\in N} d_i(e)$ represents the maximum aggregate demand for any item.
Since the total demand is bounded by the number of agents (i.e., $\sum_{i\in N} d_i(e) \le n$), this result strictly refines the general $O(\ln n)$ bound, and offers significant improvements when total demand is low.
Note that we can assume without loss of generality that $\sum_{i\in N} d_i(e) > 1$ for every item $e\in M$, as otherwise we can allocate item $e$ to satisfy all demands on the item.

\smallskip

Following the analysis from Section~\ref{sec:general_concave}, we use $x_i^\pi$ to denote the temporary bundle allocated to agent $i$ upon its arrival, under permutation $\pi$, and $\hat{x}_i = \E_\pi \left[ x_i^\pi \right]$.
Note that for capped functions, we can assume that $x_i^{\pi} \leq d_i$ holds for all $i\in N$ and permutation $\pi$.

By Lemma~\ref{lemma:concave_Shapley_upper_bound}, we can upper bound $\phi_i$ by $v_i(\hat{x}_i)$.
In the following, we show an improved upper bound on the summation of $\hat{x}_i$.

\begin{lemma} \label{lemma:upper_bound_sum_of_x_i(e)_capped}
    For any item $e\in M$, we have $\sum_{i\in N} \hat{x}_i(e) \le \ln(\sum_{i\in N} d_i(e)) + 2$.
\end{lemma}
\begin{proof}
    Following the same analysis as in Lemma~\ref{lemma:upper_bound_sum_of_x_i(e)}, we have
    \begin{align*}
        \sum_{i\in N} \hat{x}_i(e) = \sum_{i\in N} \E_\pi \left[ x_i^\pi(e) \right] 
        = \sum_{j=1}^n \E_\pi \left[ x_{\pi^{-1}(j)}^\pi(e) \right].
    \end{align*}
    
    Consider the term $\E_\pi [x_{\pi^{-1}(j)}^\pi(e)]$. 
    Conditioned on the set $S_j$ of the first $j$ agents, we have:
    \begin{equation*}
        \E \left[x_{\pi^{-1}(j)}^\pi(e) \mid S_j\right] = \frac{1}{j} \sum_{k \in S_j} x_k(e) \leq \frac{1}{j}\cdot \min\left\{1,\sum_{k\in S_j} d_k(e)\right\}.
    \end{equation*}

    Since $S_j$ is a random subset of $j$ agents, by taking the expectation over $S_j$, we have
    \begin{equation*}
        \E \left[x_{\pi^{-1}(j)}^\pi(e) \right]
        \leq \min\left\{\frac{1}{j},\frac{1}{j}\cdot \frac{j}{n}\cdot \sum_{k\in N} d_k(e)\right\}
        = \min\left\{\frac{1}{j},\frac{1}{n}\cdot \sum_{k\in N} d_k(e)\right\}.
    \end{equation*}
    
    Summing over all positions, we obtain
    \begin{align*}
        \sum_{i\in N} \hat{x}_i(e) & = \sum_{j=1}^n \min\left\{\frac{1}{j},\frac{1}{n}\cdot \sum_{k\in N} d_k(e)\right\}
        \leq 1 + \sum_{j=\left\lfloor \frac{n}{\sum_{k\in N} d_k(e)}\right\rfloor + 1}^n \frac{1}{j} \\
        & \leq 1 + \ln(n) + 1 - \ln\left( \left\lfloor \frac{n}{\sum_{k\in N} d_k(e)}\right\rfloor + 1 \right) \\
        & \leq \ln\left(\sum_{k\in N} d_k(e) \right) + 2.
        \qedhere
    \end{align*}
\end{proof}

Let $D = \max_{e\in M}\{ \sum_{i\in N} d_i(e) \}$ and $\alpha = \ln D + 2$.\footnote{Note that we define $D$ to be maximum of $\sum_{i\in N} d_i(e)$ across all items so that the same upper bound can be applied to all dimensions of the general concave function.
This is, in some sense, inevitable because there may be one item that is significantly more valuable than all other items to all agents.}
By Lemma~\ref{lemma:upper_bound_sum_of_x_i(e)_capped}, we have $\sum_{i\in N} \hat{x}_i(e) \le \alpha$ for every item $e\in M$.
For each agent $i\in N$ and each item $e\in M$, we allocate
\begin{equation*}
    x_i(e) = \frac{\hat{x}_i(e)}{\sum_{j\in N} \hat{x}_j(e)}
\end{equation*}
fraction of $e$ to agent $i$.
The allocation is feasible.
It is also $\alpha$-SVF because
\begin{equation*}
    v_i(x_i) \ge v_i \left(\frac{\hat{x}_i}{\alpha}\right) \ge \frac{v_i(\hat{x}_i)}{\alpha} \ge \frac{\phi_i}{\alpha}.
\end{equation*}



\subsection{Lower Bound for (Capped) Concave Functions}
\label{ssec:lower_bound_capped_concave}

In this section, we show a matching lower bound of $\Omega(\ln D)$ for the approximation ratio of the Shapley value for capped concave functions.
Note that when $d_i(e) = 1$ for every agent $i\in N$ and item $e\in M$, the capped concave functions become the general concave functions.
Therefore, we also have a matching lower bound of $\Omega(\ln n)$ for general concave functions.

The following theorem demonstrates that even in the restricted domain of linear valuations, a logarithmic factor is inevitable.

\begin{theorem}[Lower Bound]\label{thm:general_lower_bound}
    There exists an instance such that each agent has a capped linear valuation, in which no allocation can guarantee an approximation ratio better than $\ln D$, with respect to the Shapley value, where $D = \max_{e\in M} \{\sum_{i\in N} d_i(e)\}$.
\end{theorem}
\begin{proof}   
    We construct a hard instance with a single divisible item and $n$ agents.
    For this proof, we slightly abuse notation: let $v_i$ denote the scalar value that agent $i$ assigns to the entire item.
    Suppose that each agent has an identical demand $d \in (0,1]$ on resource $e$, and the utility agent $i$ derives from a fraction $x_i \in [0,1]$ of the item is $v_i \cdot \min \{x_i, d\}$.
    We set the values to be exponentially decreasing: $v_i = \beta^{n-i+1}$ for a large constant $\beta \geq 2$.
    Thus, $v_1 > v_2 > \dots > v_n$.

    Let $K = 1/d$ denote the maximum number of agents whose requests can be fully satisfied simultaneously (assuming that $1/d$ is an integer).
    We analyze the marginal contributions as $\beta \to \infty$. 
    In this regime, the social welfare of any coalition $S$ is maximized by greedily satisfying the demand of the members with the highest scalar values until the resource runs out.
    Consider a specific agent $i$. 
    Because valuations are exponentially separated, agent $i$ contributes to the welfare if and only if the resource has not been fully exhausted by agents with strictly higher values (i.e., agents with indices $j < i$).
    Recall that $S_i^\pi$ is the set of predecessors of $i$ under the random permutation $\pi$.
    Agent $i$ makes a (almost) full marginal contribution of $d \cdot v_i$ if and only if the number of higher-value predecessors is less than the capacity limit $K$:
    \begin{equation*}
        |\{j \in S_i^\pi : j < i\}| < K.
    \end{equation*}

    This condition depends solely on the relative ordering of the set $\{1, \dots, i\}$. 
    Agent $i$ contributes if he appears among the first $K$ positions within this specific subset of $i$ agents.
    Since the position of $i$ is uniformly distributed within $\{1, \dots, i\}$, the probability of this event is:
    \begin{equation*}
        \text{Pr}[\text{$i$ contributes}] = \min\left\{1, \frac{K}{i}\right\}.
    \end{equation*}
    Consequently, as $\beta \to \infty$, the Shapley value of agent $i$ approaches: $\phi_i \to d \cdot v_i \cdot \min\left\{1, \frac{K}{i}\right\}$.

    Now, consider an $\alpha$-SVF allocation.   
    For each agent, the allocation $x_i$ must satisfy 
    $$
    v_i \cdot \min\{x_i, d\} \ge \frac{\phi_i}{\alpha}.
    $$

    Substituting the expression for $\phi_i$, this requires:
    \begin{equation*}
        x_i \ge \frac{d}{\alpha}\cdot \min\left\{1, \frac{K}{i}\right\}.
    \end{equation*}
    
    Summing the allocation over all agents implies the feasibility constraint:
    \begin{align*}
        1 = \sum_{i\in N} x_i 
        &\ge \sum_{i=1}^n \frac{d}{\alpha} \min\left\{1, \frac{K}{i}\right\} \\
        &= \frac{d}{\alpha} \left( \sum_{i=1}^K 1 + \sum_{i=K+1}^n \frac{K}{i} \right) 
        = \frac{d\cdot K}{\alpha} \left( 1 + \sum_{i=K+1}^n \frac{1}{i} \right).
    \end{align*}

    Note that we can simplify the expression using the harmonic number approximation
    \begin{equation*}
        \sum_{i=K+1}^n \frac{1}{i} \ge \ln (n+1) - (\ln K + 1).
    \end{equation*}
    Since $d \cdot K = 1$, rearranging for $\alpha$ yields $\alpha \ge \ln \left(\frac{n+1}{K} \right)$.
    
    Finally, recall that $K = 1/d = n/D$, the approximation ratio is lower bounded by:
    \begin{equation*}
        \alpha \ge \ln \left(\frac{n+1}{n}D \right) \ge \ln D.
        \qedhere
    \end{equation*}
\end{proof}

\paragraph{Remarks.}
It is natural to ask whether randomization can circumvent the logarithmic bottleneck for ex-ante guarantees. We observe that this is not the case. Consider any randomized allocation $\bx$, and let $x'_i = \mathbb{E}[x_i]$ denote the corresponding expected allocation for agent $i$. By the feasibility of the randomized allocation and the convexity of the feasible region, $\bx'$ is feasible. Moreover, for additive valuations, linearity implies
\[
v_i(x'_i) = \mathbb{E}[v_i(x_i)].
\]
For concave valuations, Jensen's inequality gives
\[
v_i(x'_i) \geq \mathbb{E}[v_i(x_i)].
\]
Since the Shapley values are determined solely by the underlying instance and are independent of the allocation rule, the deterministic allocation $\bx'$ attains an ex-ante guarantee that is at least as strong as that of the randomized allocation. Therefore, randomization cannot improve the approximation ratio in this setting.

\subsection{Computation}
Finally, we discuss the computational aspect of the optimal Shapley approximation rule. 
Given an instance with concave valuations, computing the optimal Shapley approximation ratio and the corresponding allocation consists of two steps: 
\begin{enumerate}
    \item \textbf{Step 1:} (Approximately) compute the Shapley value $\phi_i$ for each agent $i\in N$.
    \item \textbf{Step 2:} Find an allocation $\bx$ that best approximates the Shapley value $\phi_i$ for all agents.
\end{enumerate}
We discuss each step in turn.

\paragraph{Step 1: Computing the Shapley Value.}
Although the Shapley value has a closed-form definition, computing the Shapley value exactly is generally intractable, as it requires evaluating the marginal contribution of each agent to all possible coalitions. This problem has been shown to be \#P-complete even for simple valuation functions~\cite{journals/mor/DengIN99,conitzer2004computing}.
However, for concave valuations, we can efficiently estimate the Shapley value up to any desired accuracy using random sampling.
\begin{lemma}\label{lem:shapley_estimation}
    For any instance with concave valuations, given $\epsilon, \delta \in (0,1)$, we can compute in time polynomial in $n$, $m$, $1/\epsilon$, and $\log(1/\delta)$, an estimate $\tilde{\phi}_i$ for each agent $i\in N$ such that with probability at least $1-\delta$, we have
    \begin{equation*}
        (1-\epsilon)\cdot \phi_i \le \tilde{\phi}_i \le (1+\epsilon)\cdot \phi_i.
    \end{equation*}
\end{lemma}
The key idea is to approximate the Shapley value of each agent by sampling a polynomial number of random arrival orders and averaging the agent's marginal contributions across these orders. Via concentration inequalities (e.g., Hoeffding's inequality), we can ensure that the estimate converges to the true Shapley value with high probability. Details are provided in Appendix~\ref{sec:approx_via_sample}.

\paragraph{Step 2: Finding the (Near-)Optimal Allocation.}
With the estimated Shapley values $\tilde{\phi}_i$ for each agent $i\in N$, we can formulate the problem of finding the optimal Shapley approximation allocation as a convex optimization problem.
\begin{equation}
\begin{aligned}
    \max_{\lambda, x}. \quad & \lambda \\
    \text{s.t.} \quad & v_i(x_i) \ge \lambda \cdot \tilde{\phi}_i, \quad \forall i\in N \\
    & \sum_{i \in N} x_i(e) = 1, \quad \forall e\in M \\
    & x_i(e) \ge 0, \quad \forall e\in M, i\in N
\end{aligned}
\label{eq:shapley_lp}
\end{equation}

Since the valuation functions $v_i$ are concave, the constraints $v_i(x_i) \ge \lambda \cdot \tilde{\phi}_i$ define convex sets. 
Thus, the above optimization problem is a convex program that can be solved efficiently using standard convex optimization techniques.
Putting together the two steps, we have the following.

\begin{theorem}\label{theorem:concave_computation}
    For any instance with concave valuations, given $\epsilon, \delta \in (0,1)$, we can compute in time polynomial in $n$, $m$, $1/\epsilon$, and $\log(1/\delta)$, an allocation $\bx$ such that with probability at least $1-\delta$, the allocation is a $(1+\epsilon)$-approximation of the optimal Shapley approximation ratio.
\end{theorem}


\section{Shapley Value Fairness for Linear Valuations: Refined Bounds}\label{sec:general-bounds}

A special and important class of concave functions is linear functions. 
We slightly abuse notation and let $v_i(e)$ be the utility of agent $i$ for receiving (one unit of) item $e$, which is the scalar for the valuation of agent $i$ on coordinate $e$.
The valuation function of agent $i$ can be expressed as
\begin{equation*}
    v_i(x_i) = \sum_{e\in M} x_i(e) \cdot v_i(e).
\end{equation*}

Linear functions naturally arise in many applications, e.g., when agents have constant marginal returns to scale.
One could naturally ask whether better approximation ratios can be achieved for linear functions.
However, note that the upper bound of $O(\ln n)$ in Theorem~\ref{theorem:concave_upper_bound} still holds for linear functions.
Furthermore, the hard instance provided in Section~\ref{ssec:lower_bound_capped_concave} is also linear, which means $O(\ln n)$ approximation guarantee is asymptotically tight.

In this section, we seek to refine the approximation ratio of SVF for linear valuations by considering additional structural properties of agents' valuations.
More specifically, we consider the approximation ratio of SVF parameterized by the number of agent types and the ratio between the maximum and minimum value.
Such parametrized instances have also been studied previously in the discrete fair division problem~\cite{journals/dam/Mahara23,conf/sigecom/HVGN025,conf/atal/BarmanKP24}.

We summarize our results in the following theorem.
\begin{theorem}
\label{theorem:summary_linear}
There exists a polynomial-time algorithm that computes an $O\bigl(\min\{k, \ln \gamma, \ln n\}\bigr)$-SVF allocation for linear valuations, where
\begin{itemize}
    \item $k$ is the number of agent types, i.e., there are at most $k$ distinct valuation functions among all agents; and
    \item $\gamma$ is the largest ratio between the maximum and minimum value of any item across all agents,
    that is, $\displaystyle \gamma = \max_{e \in M, i, j \in N, v_j(e) \neq 0} v_i(e) / v_j(e)$.
\end{itemize}
Moreover, the approximation ratio is asymptotically tight in all parameters $n$, $k$, and $\gamma$.
\end{theorem}



Our results show that under either of these two structural properties, the approximation ratio can be significantly improved from $O(\ln n)$ to $O(k)$ or $O(\ln \gamma)$, respectively.
This is particularly meaningful in practice, as in many applications the number of agent types is small, and agents' valuations do not vary significantly.

In the rest of this section, we first provide a characterization of the Shapley value for linear valuations, which will be useful in the subsequent analysis.
We then prove each part of Theorem~\ref{theorem:summary_linear} in separate subsections.

\subsection{Characterization of Shapley Value}
\label{ssec:computation_linear}
We first show for linear valuations, the Shapley value $\phi_i$ of each agent $i\in N$ can be concisely characterized and computed in polynomial time.


Let $y^e\in [0,1]^M$ be the indicator vector for item $e$, i.e., $y^e(e) = 1$ and $y^e(e') = 0$ for all $e' \neq e$.
Then we define the Shapley value of agent $i$ on item $e$ as
\begin{equation*}
    \phi_i(e) = \E_{\pi}\left[ f(S_i^\pi \cup \{i\}, y^e) - f(S_i^\pi, y^e) \right].
\end{equation*}

By the definition of Shapley value, for linear functions we have $\phi_i = \sum_{e\in M} \phi_i(e)$.
Therefore, to compute the Shapley value, it suffices to compute $\phi_i(e)$ for each item $e\in M$.
Fix an arbitrary item $e$ and focus on the allocation of this single item.
We assume without loss of generality (w.l.o.g.)\footnote{This is w.l.o.g. only if we fix an item.} that 
$$
v_1(e) \geq v_2(e) \geq \cdots \geq v_n(e).
$$

We further assume that $v_{n+1}(e) = 0$ and $\phi_{n+1}(e) = 0$ for notational convenience.
It can be equivalently interpreted as we introduce a dummy agent $n+1$ with zero value on item $e$.
Note that additional agents with a value of $0$ for the item will not change the Shapley value of other agents.
In the case of allocating a single item, the social welfare function $f$ is simply a max function, i.e., for any subset of agents $S$ we have 
$f(S) = \max_{i\in S} \{v_i(e)\}$.

Next, we derive a closed-form expression for $\phi_i(e)$, which will later facilitate the analysis of Shapley value fairness under different settings.
A similar expression can be derived when agents have capped linear functions with uniform demand (see Appendix~\ref{sec:computation_uniform_demand}).

\begin{lemma} \label{lemma:shapley_value_expression}
    For any agent $i \in [n]$, we have $\phi_{i}(e) = \sum_{t=i}^n {\frac{v_t(e) - v_{t+1}(e)}{t}}$.
\end{lemma}
\begin{proof}
    Recall that an agent's Shapley value represents her expected marginal contribution to the total welfare when agents arrive in a uniform random order $\pi:N\to [n]$, where $\pi(i)$ represents the arrival order of agent $i$.
    Before any agent has arrived, the social welfare is $0$; after all agents have arrived, the social welfare is $v_1(e)$.
    Therefore, the social welfare increases gradually following the arrival of agents.
    We cut the total increment $v_1(e)$ into chunks
    \begin{equation*}
        v_1(e)-v_2(e),\; v_2(e)-v_3(e),\;\ldots,\; v_{n-1}(e)-v_{n}(e),\; v_n(e)
    \end{equation*}
    and analyze which agent is responsible for the marginal increment $v_t(e) - v_{t+1}(e)$ of the social welfare under this order $\pi$.
    Note that this should be accounted for by the \emph{first} arriving agent among agents $\{1, \ldots, t\}$ under $\pi$. 
    Let $S_t = \{1, \ldots, t\}$.
    When $\pi$ is a random permutation, each of the agents in $S_t$ has the same probability $1/t$ of arriving first among $S_t$. 
    This gives each of the agents in $S_t$ a marginal contribution of $(v_t(e) - v_{t+1}(e))/t$ to their Shapley value.
    Since agent $i$ appears in $S_t$ for every $t\geq i$, by a summation over all $t\geq i$, we get
    \begin{equation*}
        \phi_{i}(e) = \sum_{t=i}^n {\frac{v_t(e) - v_{t+1}(e)}{t}}.  \qedhere
    \end{equation*}
\end{proof}

Lemma~\ref{lemma:shapley_value_expression} implies the following recursive form for $\phi_i(e)$ (recall that $\phi_{n+1}(e) = 0$).

\begin{corollary} \label{corollary:shapley_value_recursive}
    For any agent $i \in [n]$, we have
    \begin{equation*}
        \phi_i(e) = \frac{v_i(e) - v_{i+1}(e)}{i} + \phi_{i+1}(e).
    \end{equation*}
\end{corollary}

Following Corollary~\ref{corollary:shapley_value_recursive}, for each fixed item $e$, we can compute the Shapley value (on $e$) for all agents in $O(n)$ time.
This gives an $O(mn)$ time algorithm for computing the Shapley value for all agents on the items $M$. Lastly, notice that the optimization problem~\ref{eq:shapley_lp}, with exact Shapley values and linear valuations, becomes a linear program, which can be solved in polynomial time.

Finally, we observe that with linear valuations, it is also w.l.o.g. to assume that $M$ contains only one item $e$ when analyzing the approximation ratio of SVF.
This is because both the valuation functions and Shapley values are linear across items.


\begin{observation}\label{obs:single_item_reduction}
    If there exists an allocation $\bx$ such that $v_i(x_i(e)) \geq \frac{\phi_i(e)}{\alpha}$ for any agent $i\in N$ and item $e\in M$, then $\bx$ is $\alpha$-SVF.
\end{observation}

\subsection{Few Agent Types}
 
In this section, we consider the setting when there are only a few types of valuations.
While the preceding analysis yields an asymptotic bound of $\Theta(\ln n)$, this result is driven by an extreme instance in which agents are highly heterogeneous: 
each agent possesses a unique valuation, and these valuations differ exponentially.
We now consider the case where agents are grouped into $k \ll n$ distinct types, thereby relaxing the heterogeneity assumption.
We show that when there are at most $k$ agent types, there exists an allocation that is $k$-SVF.
Moreover, the approximation ratio is asymptotically tight when $k=o(\ln n)$.\footnote{Since we can compute $(\ln n + 1)$-SVF for general instances, it is meaningless to consider $k = \Omega(\ln n)$.}

\smallskip

Recall that we assume w.l.o.g. that $M = \{e\}$ and $\phi_i = \phi_i(e)$.
For convenience of notation, we use $v_i$ to denote $v_i(e)$, $x_i$ to denote $x_i(e)$.
The utility of agent $i$ for receiving $x_i$ fraction of item $e$ is given by $x_i\cdot v_i$.  
Suppose there are $k$ distinct values $v^1 > v^2 > \cdots > v^k$.
Let $N = N_1\cup N_2\cup \cdots \cup N_k$, where $N_i$ denotes the set of agents of type-$i$, and every agent $j\in N_i$ satisfies $v_j = v^i$.
Let $q_i = |N_i|$ denote the number of agents of type-$i$.

\begin{theorem} \label{theorem:k_types}
For any instance with $k$ agent types, there exists a polynomial-time algorithm that computes $k$-SVF allocations.
Moreover, the ratio is asymptotically tight for $k=o(\ln n)$.
\end{theorem}

To prove the above theorem, we first propose an algorithm for the computation of $k$-SVF allocations.
Note that it suffices to consider the case where $k = o(\ln n)$; otherwise, we have $O(\ln n) = O(k)$ and the general bound already applies.
    
Fix any agent $j \in N_i$.
We have $v_j = v^i$.
Recall that all agents in $N_1,N_2,\ldots,N_i$ have value at least $v^i$ on item $e$.
Therefore agent $j$ contributes a non-zero marginal value only if she appears before all other agents of types in $\{1,\ldots,i\}$.
Since there are $\sum_{t\leq i} q_t$ agents of types in $\{1,\ldots,i\}$, and the marginal contribution of agent $j$ is at most $v^i$, we have
\begin{equation*}
    \phi_j \le \frac{v^i}{\sum_{t=1}^{i} q_t}.
\end{equation*}

For any agent $i\in N$, we use $\tau(i) \in [k]$ to denote its type.
To ensure an $\alpha$-SVF allocation, we assign to each agent $i$ a fraction
\begin{equation*}
    x_i = \frac{1}{\alpha \cdot \sum_{t=1}^{\tau(i)} q_t},
    \quad\text{where}\quad
    \alpha = \sum_{j\in N}\frac{1}{\sum_{t=1}^{\tau(j)} q_t}.
\end{equation*}

It is straightforward to verify that the allocation is complete:
\begin{equation*}
    \sum_{i\in N} x_i = \frac{1}{\alpha} \cdot \sum_{i\in N} \frac{1}{\sum_{t=1}^{\tau(i)} q_t} = 1.
\end{equation*}

Finally, we derive an upper bound on~$\alpha$.
Grouping agents by type, since every $j\in N_i$ has the same denominator $\sum_{t=1}^i q_t$, we obtain
\begin{equation*}
    \alpha=\sum_{j\in N}\frac{1}{\sum_{t=1}^{\tau(j)} q_t}
    =\sum_{i=1}^k \sum_{j\in N_i}\frac{1}{\sum_{t=1}^{i} q_t}
    =\sum_{i=1}^k \frac{q_i}{\sum_{t=1}^{i} q_t} \leq k.
\end{equation*}

Hence, the approximation ratio of SVF is at most $k$.

It remains to show that there exists an instance with $k=o(\ln n)$ agent types, for which no allocation can guarantee an approximation ratio of $o(k)$ with respect to SVF.

Consider an instance with $n \ge 2^k$ agents, where agents with index in the interval $[2^{i-1},2^i)$ are of the same type, and have value $v^i = 2^{k-i}$ for each $i \in [k]$. 
In other words, agent $1$ has value $2^{k-1}$; agents $2,3$ has value $2^{k-2}$; agents $4, 5, 6, 7$ has value $2^{k-3}$; etc.
In particular, for any agent $j \ge 2^k$, we have $v_j = v^k = 1$. 
Observe that there are exactly $k$ distinct valuations.

For all $t < k$, the number of type-$t$ agents is $q_t = |[2^{t-1},2^t)| = 2^{t-1}$.
For type-$k$ agents (with value $1$), we have
\begin{equation*}
    q_k = n - \sum_{t<k} q_t
    = n - 2^{k-1} + 1.
\end{equation*}

Since type-$k$ agents have the smallest valuation, 
a type-$k$ agent has a non-zero marginal contribution only if she arrives first among all agents, 
which occurs with probability $1/n$.

Now, fix any $t < k$ and consider type-$t$ agents, each with value $2^{k-t}$.
The total number of agents with value at least $2^{k-t}$ is 
$
1 + 2 + \cdots + 2^{t-1} = 2^t - 1
$.
A type-$t$ agent has a non-zero marginal contribution only if she arrives first among these agents, 
which occurs with probability $1/(2^t - 1)$.
In summary, the probability that agents of type-$t$ have non-zero marginal contribution is
\begin{equation*}
    p(t) =
    \begin{cases}
    \dfrac{1}{2^t - 1}, & t \in [k-1], \\[8pt]
    \dfrac{1}{n}, & t = k.
    \end{cases}
\end{equation*}

Observe that if an agent $i$ has a non-zero marginal contribution, then this contribution is at least $v_i/2$. 
This follows from the construction of distinct values: 
for $i < k$, we have $v^{i} - v^{i+1} = v^{i}/2$; and for the lowest level, we have $v^{k} - 0 = v^{k} \ge v^{k}/2$.
Consequently, for any agent $i$ of type~$t$,
\begin{equation*}
    \phi_i \ge 
    \begin{cases}
    \dfrac{v_i}{2} \cdot \dfrac{1}{2^t - 1}, & t \in [k-1],\\[8pt]
    \dfrac{v_i}{2} \cdot \dfrac{1}{n}, & t = k.
    \end{cases}
\end{equation*}

To achieve $\alpha$-SVF, the total allocation must satisfy
\begin{equation*}
    1 = \sum_{i\in N} x_i \geq \sum_{i\in N} \frac{\phi_i}{v_i \cdot \alpha} = \frac{1}{\alpha} \cdot \sum_{i\in N}\frac{\phi_i}{v_i},
    \quad\text{which implies}\quad
    \alpha \ge \sum_{j\in N} \frac{\phi_j}{v_j}.
\end{equation*}

Grouping agents by their types and using $q_i$, we obtain
\begin{equation*}
    \sum_{j\in N} \frac{\phi_j}{v_j}
    \ge
    \sum_{t=1}^{k-1} \left( q_t \cdot \frac{1}{2} \cdot \frac{1}{2^t - 1} \right)
    + q_k \cdot \frac{1}{2} \cdot \frac{1}{n}
    =
    \sum_{t=1}^{k-1} \left( \frac{1}{2} \cdot \frac{2^{t-1}}{2^t - 1} \right)
    + \frac{q_k}{2n}.
\end{equation*}

Since $\frac{2^{t-1}}{2^t - 1} \ge \frac{1}{2}$ for all $t \ge 1$ and $q_k \geq n/2$, we obtain
\begin{equation*}
    \alpha 
    \ge
    \sum_{j\in N} \frac{\phi_j}{v_j}
    \ge
    \frac{k}{4}.
\end{equation*}

Therefore, the approximation ratio must be at least $\Omega(k)$, matching the upper bound.


\subsection{Bounded Valuations}

In this section, we consider the setting where the maximum and minimum values of agents for any item are bounded within a multiplicative factor $\gamma$.
Following the same set of assumptions and notations as in the previous sections, we assume that there is only one item $e$, and use $v_i$ to denote the value of agent $i$, where
\begin{equation*}
    v_1 \geq v_2 \geq \cdots \geq v_n > 0\footnote{Agents with zero value for the item have no contribution and can be ignored from the analysis.}.
\end{equation*}
and these valutaions are bounded within a multiplicative factor $\gamma$, i.e., $v_1 / v_n \leq \gamma$.
This factor $\gamma$ captures the variation in agents' valuations. In many real-world applications, it is reasonable to assume that agents' valuations do not vary significantly. For example, in the research funding allocation problem discussed in Section~\ref{sec:introduction}, proposals are typically evaluated based on a common set of criteria, such as scientific merit, feasibility, and potential impact. The difference in scores between the highest-rated and lowest-rated proposals is often within a small constant factor, reflecting the relative quality of the proposals rather than extreme disparities.

We show that under this condition, the optimal approximation ratio is always $O(\ln \gamma)$, and this ratio is also asymptotically tight with regard to $\gamma$ (when $\gamma = o(n)$).

\begin{theorem} \label{theorem:gamma_bounded}
For any instance in which $\gamma = \max_{e \in M, i, j \in N, v_j(e) \neq 0} v_i(e) / v_j(e)$, there exists an $O(\ln\gamma)$-SVF allocation which can be computed in polynomial time.
Moreover, this approximation factor is asymptotically tight with regard to $\gamma$ when $\gamma = o(n)$.
\end{theorem}
\begin{proof}
By Observation~\ref{obs:single_item_reduction}, it suffices to consider the case where there is only one item $e$.
We first establish an upper bound on the approximation ratio~$\alpha$ by presenting an algorithm for computing the allocation.
Observe that if each agent $i$ receives a fraction $x_i = \frac{\phi_i}{v_i \cdot \alpha}$ of the item for some $\alpha$, then the allocation is naturally $\alpha$-SVF. 
The completeness condition requires $\sum_{i\in N} x_i = 1$, which is equivalent to 
$\alpha = \sum_{i\in N} \frac{\phi_i}{v_i}$.
Therefore, to obtain an upper bound on $\alpha$, 
it suffices to bound the quantity $\sum_{i\in N} \frac{\phi_i}{v_i}$.
As before, we introduce $v_{n+1} = \phi_{n+1} = 0$ for convenience of analysis.

\paragraph{Upper Bound}
From the expression of the Shapley values (Lemma~\ref{lemma:shapley_value_expression}), we have 
$\phi_i = \sum_{t=i}^{n} \frac{v_t - v_{t+1}}{t}$.
Substituting the expression into our expression for $\alpha$ and further exchanging the order of summation yields:
\begin{align*}
    \sum_{i=1}^{n} \frac{\phi_i}{v_i} = 
    \sum_{i=1}^{n} \sum_{t=i}^{n} \frac{v_t - v_{t+1}}{t \cdot v_i}
    = \sum_{t=1}^{n} \frac{v_t - v_{t+1}}{t} \left( \sum_{i=1}^{t} \frac{1}{v_i} \right).
\end{align*}

Since the valuations are non-increasing, we have $v_i \ge v_t$ for all $i \le t$, which implies $1/v_i \le 1/v_t$. 
Consequently, we can bound the inner sum as $\sum_{i=1}^{t} \frac{1}{v_i} \le \frac{t}{v_t}$.
Applying this bound simplifies the total summation:
\begin{equation*}
    \sum_{i=1}^{n} \frac{\phi_i}{v_i} \le
    \sum_{t=1}^{n} \frac{v_t - v_{t+1}}{t} \cdot \frac{t}{v_t} =
    \sum_{t=1}^{n} \left( 1 - \frac{v_{t+1}}{v_t} \right).
\end{equation*}

To finalize the bound, we separate the last term (where $v_{n+1}=0$) and apply the inequality $1-x \le \ln(1/x)$ for $x \in (0,1]$ to the first $n-1$ terms:
\begin{align*}
    \sum_{i=1}^{n} \frac{\phi_i}{v_i}
     & \le \underbrace{\sum_{t=1}^{n-1} \Bigl(1 - \frac{v_{t+1}}{v_t}\Bigr)}_{\text{Using }1-x\le \ln(1/x)}
     +  \Bigl(1 - \frac{v_{n+1}}{v_n}\Bigr) \\
     & \le \sum_{t=1}^{n-1} \ln\!\left( \frac{v_t}{v_{t+1}} \right) + 1  = \ln\left(\frac{v_1}{v_2} \cdot \frac{v_2}{v_3} \cdots \frac{v_{n-1}}{v_n}\right) + 1 \\
     & = \ln\left( \frac{v_1}{v_n} \right) + 1 = \ln \gamma + 1.
\end{align*}

Therefore, by allocating a $\phi_i/(v_i\cdot \alpha)$ fraction to each agent $i$, where $\alpha = \sum_{i\in N} (\phi_i/v_i)$, we obtain a $(\ln \gamma + 1)$-SVF allocation.

\paragraph{Lower Bound}
Finally, we demonstrate the tightness of this bound.
For any fixed integer parameter $\gamma \le n$, consider an instance defined as follows:
\begin{equation*}
    v_i =
    \begin{cases}
        \frac{\gamma}{i}, & i\le \gamma,\\[4pt]
        1, & i>\gamma.
    \end{cases}
\end{equation*}

Note that 
$$\gamma = v_1 > v_2 > \cdots > v_\gamma = v_{\gamma+1} = \cdots = v_n = 1,$$
and $v_1/v_n = \gamma$ is satisfied.
Recall that to achieve an $\alpha$-SVF allocation, the approximation ratio must satisfy the feasibility constraint:
\begin{equation*}
    1 = \sum_{i\in N} x_i \geq \sum_{i\in N} \frac{\phi_i}{v_i \cdot \alpha} = \frac{1}{\alpha} \cdot \sum_{i\in N}\frac{\phi_i}{v_i},
    \quad\text{which implies}\quad
    \alpha \ge \sum_{j\in N} \frac{\phi_j}{v_j}.
\end{equation*}

We invoke the equation derived in the upper bound analysis:
\begin{equation*}
    \sum_{i=1}^{n}\frac{\phi_i}{v_i} = \sum_{t=1}^{n} \frac{v_t-v_{t+1}}{t} \left( \sum_{i=1}^{t}\frac{1}{v_i} \right).
\end{equation*}

To establish the lower bound, it suffices to sum over the subset of indices $t \in \{1, \dots, \lfloor \gamma \rfloor - 1\}$. 
For any $t$ in this range, both $t$ and $t+1$ are at most $\gamma$, so the values follow the form $v_t = \gamma/t$. 
We compute the two components of the summand:
\begin{enumerate}
    \item \textbf{Value difference:}
    \begin{equation*}
        v_t - v_{t+1} = \frac{\gamma}{t} - \frac{\gamma}{t+1} = \frac{\gamma}{t(t+1)}.
    \end{equation*}
    \item \textbf{Sum of inverse values:}
    \begin{equation*}
        \sum_{i=1}^{t} \frac{1}{v_i} = \sum_{i=1}^{t} \frac{i}{\gamma} = \frac{1}{\gamma} \cdot \frac{t(t+1)}{2} = \frac{t(t+1)}{2\cdot \gamma}.
    \end{equation*}
\end{enumerate}

Substituting these back into the main summation, the expression becomes:
\begin{align*}
    \sum_{i=1}^{n}\frac{\phi_i}{v_i} 
    &\ge \sum_{t=1}^{\lfloor \gamma \rfloor - 1} \frac{1}{t}\cdot \left( \frac{\gamma}{t(t+1)} \right) \left( \frac{t(t+1)}{2\cdot \gamma} \right) \\
    &= \sum_{t=1}^{\lfloor \gamma \rfloor - 1} \frac{1}{2t}
    = \frac{1}{2}\cdot H_{\lfloor \gamma \rfloor - 1} = \Omega(\ln \gamma).
\end{align*}

Therefore, any allocation has an approximation ratio of SVF at least $\Omega(\ln \gamma)$.
\end{proof}



\section{Conclusion}
In this paper, we study the problem of allocating resources among agents with heterogeneous capabilities and potential contributions to social welfare.
We observe that in many practical settings (e.g., research budget distribution), classical axioms based on symmetry may not fully capture the efficiency differences between agents.
To bridge this gap, we introduce \emph{Shapley value fairness}, which derives fairness entitlements from each agent's potential contribution.
This approach extends the scope of traditional fair division, offering an alternative that better aligns with aggregate efficiency.

Several promising directions remain for future research.
Extending this framework to \emph{indivisible} items would be a natural theoretical progression, as the combinatorial nature of such problems presents new challenges for the Shapley value fairness.
It would also be interesting to identify \emph{structural properties} (analogous to the agent types and value fluctuation) that allow for better approximations in broader classes of concave valuations.

\section*{Acknowledgments}
Xiaohui Bei is partly supported by the Ministry of Education, Singapore, under its Academic Research Fund Tier 1 (RG98/23). Xiaowei Wu is funded by the University of Macau (file no. MYRG-GRG2025-00033-IOTSC), and the Science and Technology Development Fund (FDCT), Macau SAR (file no. 0147/2024/RIA2, 001/2024/SKL, 0002/2025/EQP and CG2026-IOTSC).

\newpage
\bibliographystyle{abbrv}
\bibliography{EC2026/shapley}

\newpage
\appendix

\section{Near-optimal Approximation of SVF via Sampling}
\label{sec:approx_via_sample}

In the following, we show that we can estimate $\phi_i$ for every agent $i$ up to a factor of $(1+\epsilon)$ with high probability using polynomial number of samples.
Using this estimation, we can compute an allocation with near-optimal approximation ratio to the Shapley value in polynomial time.

In the following, we fix an agent $i\in N$.
Recall that $\phi_i$ is the expected marginal contribution of agent $i$ over a random permutation $\pi$, where the marginal contribution is given by
\begin{equation*}
    \Delta(\pi) := f(S_i^\pi\cup \{i\}) - f(S_i^\pi).
\end{equation*}

We assume that given any subset $S$ of agents, we can compute the allocation that maximizes social welfare for the coalition $S$ in polynomial time, e.g., by solving a convex program.
Therefore, for every fixed permutation $\pi$, we can compute $\Delta(\pi)$ in polynomial time.

Let $z_1, z_2, \ldots, z_T$ be $T$ independent samples of $\Delta$ (by taking independent samples of $\pi$).
Note that every random variable $z_j$ takes values in $[0,v_i(\mathbf{1})]$, and we have $\E[z_j] = \phi_i$.
Let $\tilde{\phi_i} = \frac{1}{T}\cdot \sum_{j=1}^T z_j$ be the estimate for $\phi_i$.
Note that $\E[\tilde{\phi_i}] = \phi_i$.
We show that for sufficiently large $T$, $\tilde{\phi_i}$ is a good estimate.

Fix small constants $\epsilon, \delta > 0$  and let $T = \frac{9n^2 (\ln 2n - \ln \delta)}{2\epsilon^2}$.

\begin{lemma} \label{lemma:estimate_phi_via_sample}
    With probability at least $1-\frac{\delta}{n}$, we have
    \begin{equation*}
        (1-\epsilon/3)\cdot \phi_i \leq \tilde{\phi_i} \leq (1+\epsilon/3)\cdot \phi_i.
    \end{equation*}
\end{lemma}
\begin{proof}
    Recall that $z_1,\ldots,z_T$ are i.i.d. random variables taking values in $[0,v_i(\mathbf{1})]$.
    Furthermore, for every $j$ we have
    \begin{equation*}
        \E[z_j] = \phi_i \geq \frac{v_i(\mathbf{1})}{n},
    \end{equation*}
    where the inequality holds since when agent $i$ arrives first (which happens with probability $1/n$), her marginal contribution to the empty coalition is $v_i(\mathbf{1})$.
    Then by Hoeffding's Inequality~\cite{hoeffding1963probability}, we have
    \begin{align*}
        \Pr\left[\tilde{\phi_i} \geq (1+\epsilon/3)\cdot \phi_i\right]
        & = \Pr\left[ \sum_{j=1}^T z_j \geq (1+\epsilon/3)T\cdot \phi_i\right] \\
        & < \exp\left(-\frac{2(\epsilon/3\cdot T\cdot \phi_i)^2}{T \cdot (v_i(\mathbf{1}))^2}\right) \\
        & \leq \exp\left(-\frac{2 \epsilon^2 \cdot T}{9n^2}\right) = \frac{\delta}{2n}.
    \end{align*}

    Similarly, we also have
    \begin{align*}
        \Pr\left[\tilde{\phi_i} \leq (1-\epsilon/3)\cdot \phi_i\right]
        = \Pr\left[ \sum_{j=1}^T z_j \leq (1-\epsilon/3)T\cdot \phi_i\right] < \frac{\delta}{2n}.
    \end{align*}

    A union bound yields the lemma.
\end{proof}

Then, by a union bound over all $n$ agents, we have the following.

\begin{corollary}\label{cor:estimate_phi_via_sample}
    We can obtain estimates $\tilde{\phi_i}$ for every agent $i$ in polynomial time, such that with probability at least $1 - \delta$, for every agent $i$ we have
    \begin{equation*}
        (1-\epsilon/3)\cdot \phi_i\leq \tilde{\phi_i} \leq (1+\epsilon/3)\cdot \phi_i.
    \end{equation*}    
\end{corollary}

Next, we compute an allocation based on the estimates and show that it achieves a near-optimal approximation to the accurate Shapley value with high probability.


Let $(x_1, x_2, \ldots, x_n, \lambda)$ be the optimal solution to the following convex program.
\begin{align*}
    \mathsf{CP}(\tilde{\phi_1},\tilde{\phi_2},\ldots,\tilde{\phi_n}): \quad \max_{x,\lambda}. \qquad & \lambda \\
    s.t. \qquad & v_i(x_i) \geq \lambda\cdot \tilde{\phi_i}, \quad \forall i\in N \\
    & \sum_{i\in N} x_i(e) = 1, \quad \forall e\in M \\
    & x_i(e) \geq 0, \quad \forall i\in N, \forall e\in M
\end{align*}

Let $\alpha^*$ be the optimal approximation ratio to the (accurate) Shapley value, which is given by allocation $(x^*_1, x^*_2, \ldots, x^*_n)$.
That is, we have $v_i(x^*_i) \geq \frac{\phi_i}{\alpha^*}$ for every agent $i\in N$.

\begin{lemma} \label{lemma:near_optimal_approx_ratio}
    With probability at least $1 - \delta$, $(x_1, x_2, \ldots, x_n)$ is a $(1+\epsilon) \alpha^*$-SVF allocation.
\end{lemma}
\begin{proof}
    With probability at least $1 - \delta$, we have $(1-\epsilon/3)\cdot \phi_i\leq \tilde{\phi_i} \leq (1+\epsilon/3)\cdot \phi_i$ for every agent $i$.
    Condition on this, we show that
    \begin{equation*}
        \left( x^*_1,x^*_2,\ldots,x^*_n,\frac{1}{(1+\epsilon/3)\cdot \alpha^*} \right)
    \end{equation*}
    is a feasible solution to $\mathsf{CP}(\tilde{\phi_1},\ldots,\tilde{\phi_n})$.
    Clearly, the second and third sets of constraints of the program are satisfied.
    The first set of constraints is also satisfied because for every agent $i$ we have
    \begin{equation*}
        v_i(x^*_i) \geq \frac{\phi_i}{\alpha^*} \geq \frac{\tilde{\phi_i}}{(1+\epsilon/3)\alpha^*}.
    \end{equation*}

    Since $(x_1, x_2, \ldots, x_n, \lambda)$ is the optimal solution to the program, we have $\lambda \geq \frac{1}{(1+\epsilon/3)\alpha^*}$.
    Then, by the feasibility of the solution, for every agent $i$, we have
    \begin{equation*}
        v_i(x_i) \geq \lambda\cdot \tilde{\phi_i} \geq \frac{(1-\epsilon/3)\phi_i}{(1+\epsilon/3)\alpha^*}
        \geq \frac{\phi_i}{(1+\epsilon)\alpha^*}.
    \end{equation*}

    Therefore, $(x_1, x_2, \ldots, x_n)$ is a $(1+\epsilon) \alpha^*$-SVF allocation.
\end{proof}

\section{Computation of the Shapley Value for Linear Functions with Uniform Demand}
\label{sec:computation_uniform_demand}

When agents have capped linear functions, by linearity of Shapley value, it suffices to consider the case of allocating a single item, i.e., $M = \{e\}$.
As in Section~\ref{ssec:computation_linear}, we use $v_i$ to denote the value of agent $i$ for receiving item $e$, and assume that $v_1 \geq v_2 \geq \cdots v_n$.
We further assume that $v_{n+1} = \phi_{n+1} = 0$.
Suppose that all agents have a demand of $d\in (0,1]$ for the item, where $d > 1/n$ (otherwise we have $\phi_i = d\cdot v_i$).

We provide the following characterization for the Shapley value in this setting, which generalizes Lemma~\ref{lemma:shapley_value_expression} (special case when $d=1$).


\begin{theorem}\label{thm:shapley_share_B}
    For any agent $i \in [n]$, we have
    \begin{equation*}
        \phi_i = \sum_{t = i}^n \frac{v_t - v_{t+1}}{\max\{t, 1/d\}}.
    \end{equation*}
\end{theorem}
\begin{proof}
    We first consider the case where $B: = 1/d$ is a positive integer.
    Fix any agent $i \in [n]$.
    We analyze the expected marginal contribution of agent $i$ under a uniformly random arrival order.
    Three distinct cases may occur:
    \begin{itemize}
        \item[\textbf{(1)}] 
        The current $B$-th largest agent among the previously arrived agents is $t<i$ (so $v_t \ge v_i$). 
        In this case, the arrival of agent $i$ does not alter the set of top-$B$ agents, and its marginal contribution is $0$.
    
    \item[\textbf{(2)}]
        The current $B$-th largest agent among the previously arrived agents is $t > i$ (so $v_t \le v_i$). 
        Note that since $v_t$ is the $B$-th largest and $i < t$ has not arrived, we have $t\geq B+1$.
        The addition of agent~$i$ replaces agent~$t$ among the top-$B$ agents, resulting in a welfare gain of $(v_i - v_t)/B$.  
    
        To compute the probability of this event, we restrict attention to the relative order of the set of agents $T = \{1,2,\ldots,t\}$ (which contains both agent $i$ and $t$).
        The event occurs if and only if
        \begin{itemize}
            \item among agents in $T$, agent~$t$ appears as one of the first $B$ arrivals, and
            \item agent~$i$ appears at position $B+1$ among agents in $T$.
        \end{itemize}     
        In a random permutation of agents in $T$, this occurs with probability $\frac{B}{t(t-1)}$.
        Hence, the expected contribution of agent~$i$ in this case equals $\frac{v_i - v_t}{t(t-1)}$.
        
    \item[\textbf{(3)}]
        Agent~$i$ itself is among the first $B$ arrivals, which happens with probability $B/n$.  
        In this event, agent~$i$ is fully funded (the demand of $1/B$) and contributes $v_i/B$ to the total welfare, yielding an expected contribution of $\frac{v_i}{n}$.
    \end{itemize}

    Combining all cases, the Shapley value of agent $i$ is
    \begin{equation}
        \phi_i  = \frac{v_i}{n} + \sum_{t=\max\{i,B\}+1}^{n} \frac{v_i - v_t}{t(t-1)}.
        \label{equation:shapley_share_B_prob_form}
    \end{equation}

    We next simplify the expression.
    Using Equation~\eqref{equation:shapley_share_B_prob_form} for both $\phi_i$ and $\phi_{i+1}$ (and noting that it also applies to $\phi_{n+1}$ since $v_{n+1}=0$), 
    we obtain, for $i \ge B$,
    \begin{align*}
        \phi_i - \phi_{i+1} & = \frac{v_i - v_{i+1}}{n} + \sum_{t= i+1}^n \frac{v_i - v_t}{t(t-1)} - \sum_{t= i+2}^n \frac{v_{i+1} - v_t}{t(t-1)} \\
        & = \frac{v_i - v_{i+1}}{n} + \sum_{t= i+1}^n \frac{v_i - v_{i+1}}{t(t-1)} \\
        & = \frac{v_i - v_{i+1}}{i}.
    \end{align*}
    Similarly, for $i < B$ we obtain
    \begin{align*}
        \phi_i - \phi_{i+1} & = \frac{v_i - v_{i+1}}{n} + \sum_{t= B+1}^n \frac{v_i - v_t}{t(t-1)} - \sum_{t= B+1}^n \frac{v_{i+1} - v_t}{t(t-1)} \\
        & = \frac{v_i - v_{i+1}}{n} + \sum_{t= B+1}^n \frac{v_i - v_{i+1}}{t(t-1)} \\
        & = \frac{v_i - v_{i+1}}{B}.
    \end{align*}
    
    Combining the two cases, we obtain
    \begin{equation*}
        \phi_i = \frac{v_i - v_{i+1}}{\max\{i,B\}} + \phi_{i+1}
        = \sum_{t = i}^n \frac{v_t - v_{t+1}}{\max\{t,B\}}.
    \end{equation*}
    
    The above expression naturally extends to non-integral $1/d$ by linear interpolation.

    Suppose $1/d \notin \mathbb{Z}$, and let $B = \lfloor 1/d \rfloor$ and $\theta = 1/d - B \in (0,1)$.
    In other words, we have $1/d = B + \theta$ and $B\cdot d < 1 < (B+1)\cdot d$.
    Note that $\phi_i$ is the expected marginal contribution of agent $i$ under a random permutation, and this marginal contribution is linear in the total supply when the supply lies in $(B\cdot d, (B+1)\cdot d)$.
    Therefore, we have
    \begin{equation*}
        \phi_i = (1-\theta)\cdot \phi'_i + \theta\cdot \phi''_i,
    \end{equation*}
    where $\phi'_i$ (resp. $\phi''_i$) is the Shapley value of agent $i$ when the supply is $B\cdot d$ (resp. $(B+1)\cdot d$).
    Note that to compute $\phi'_i$, we can equivalently consider an instance with $d' = 1/B$, total supply $1$, and $v'_t = B\cdot d\cdot v_t$ for all $t$.
    Therefore by our previous analysis for integral $1/d$, we obtain
    \begin{equation*}
        \phi'_i = \sum_{t = i}^n \frac{B\cdot d\cdot (v_t - v_{t+1})}{\max\{t,B\}} 
        \quad \text{ and } \quad
        \phi''_i = \sum_{t = i}^n \frac{(B+1)\cdot d\cdot (v_t - v_{t+1})}{\max\{t,B+1\}}.
    \end{equation*}

    Then we have
    \begin{equation*}
        \phi_i = \sum_{t = i}^n \left( d\cdot (v_t - v_{t+1}) \cdot \left( \frac{1-\theta}{\max\{t/B,1\}} + \frac{\theta}{\max\{t/(B+1),1\}} \right) \right).
    \end{equation*}

    Observe that for $t < 1/d$, we have $t\leq B$ and
    \begin{equation*}
        \frac{1-\theta}{\max\{t/B,1\}} + \frac{\theta}{\max\{t/(B+1),1\}} = 1.
    \end{equation*}

    For $t > 1/d$, we have $t\geq B+1$ and
    \begin{equation*}
        \frac{1-\theta}{\max\{t/B,1\}} + \frac{\theta}{\max\{t/(B+1),1\}} = \frac{B(1-\theta) + (B+1)\theta}{t} = \frac{1}{t\cdot d}.
    \end{equation*}

    Therefore, we have
    \begin{equation*}
        \phi_i = \sum_{t = i}^n \frac{v_t - v_{t+1}}{\max\{t,1/d\}}.
        \qedhere
    \end{equation*}
\end{proof}

\end{document}